\documentclass[11pt]{article}
\usepackage{fullpage,amsfonts,amssymb,amsmath,amsthm,float,graphicx,todonotes}
\usepackage[breaklinks]{hyperref}

\newtheorem{theorem}{Theorem}
\newtheorem{lemma}[theorem]{Lemma}
\newtheorem{corollary}[theorem]{Corollary}

\theoremstyle{definition}
\newtheorem{definition}[theorem]{Definition}

\newcommand\lr[1]{\left( #1 \right)}

\newcommand\lrv[1]{\left|  #1 \right|}

\newcommand\lrb[1]{\left\lbrace #1 \right\rbrace}

\newcommand{\mb}[1]{\mathbf{#1}}

\newcommand{\tdots}{\,..\,}

\def\A{{\cal A}}

\newcommand{\BC}{\lrb{0,1}}

\newcommand{\mbb}[1]{\mathbb{#1}}

\newcommand{\CS}{_{\mathrm {CS}}}

\begin{document}

\title{Almost quadratic gap between partition complexity and query/communication complexity}
\date{}
\author{Andris Ambainis$^1$
\and Martins Kokainis$^1$
}
\maketitle

\begin{abstract}
\noindent
We show nearly quadratic separations between two pairs of complexity measures:
\begin{itemize}
\item
We show that there is a Boolean function $f$ with $D(f)=\Omega((D^{sc}(f))^{2-o(1)})$ where $D(f)$ is 
the deterministic query complexity of $f$ and $D^{sc}$ is the subcube partition complexity of $f$;
\item
As a consequence, we obtain that there is $f(x, y)$ such that $D^{cc}(f)=\Omega(\log^{2-o(1)}\chi(f))$ 
where $D^{cc}(f)$ is the deterministic 2-party communication complexity of $f$ (in the standard 2-party model
of communication) and $\chi(f)$ is the partition number of $f$.
\end{itemize}
Both of those separations are nearly optimal: it is well known that $D(f)=O((D^{sc}(f))^{2})$
and $D^{cc}(f)=O(\log^2\chi(f))$. 
\end{abstract}

\footnotetext[1]{Faculty of Computing, University of Latvia.
E-mail: {\tt andris.ambainis@lu.lv, martins.kokainis@lu.lv}.
Supported by the European Commission FET-Proactive project QALGO, ERC Advanced Grant MQC and Latvian State Research programme NexIT project No.1.}


\setcounter{page}{0}
\thispagestyle{empty}

\newpage
\section{Introduction}

Both query complexity and communication complexity of a function $f$ can be lower bounded by
an appropriate measure of complexity on partitions of the input set 
with the property that $f$ is constant on every part of the partition.

In the communication complexity setting, the {\em partition number} of $f:X\times Y\rightarrow \{0, 1\}$ 
(denoted by $\chi(f)$) is the smallest number of rectangles $X_i \times Y_i$ in a partition of $X\times Y$ 
with the property that $f$ is constant (either 0 or 1) on every $X_i \times Y_i$.
If a deterministic communication protocol communicates $k$ bits, it specifies 
a partition into $2^k$ rectangles. Therefore \cite{Yao79}, $D^{cc}(f)\geq \log \chi(f)$ where $D^{cc}(f)$ 
is the deterministic communication complexity of $f$ in the standard two-party model.  

The corresponding notion in query complexity is the {\em subcube partition complexity} $D^{sc}(f)$
of a function $f:\{0, 1\}^n\rightarrow \{0, 1\}$ \cite{FKW,Kothari}. 
It is defined as the smallest $k$ for which $\{0, 1\}^n$ can
be partitioned into subcubes $S_i$ so that each $S_i$ defined by fixing at most $k$ variables and, on each $S_i$, $f$ is a constant.
Again, it is easy to see that $D^{sc}(f)$ is a lower bound for the deterministic decision tree complexity $D(f)$
(since any deterministic decision tree defines a partition of $\{0, 1\}^n$).

Although we do not consider randomized complexity in this paper, we note that 
both of these measures have randomized counterparts \cite{FKW,JK,JLV} which provide lower bounds for randomized communication
complexity and randomized query complexity.

We study the question: how tight are the partition lower bounds? 

For communication complexity, Aho et al. \cite{AUY83} showed that $D^{cc}(f)=O(\log^2 \chi(f))$, by observing that $D(f)$ is upper-bounded by
the square of the non-deterministic communication complexity which, in turn, is at most $\log \chi(f)$.
Since then, it was an open problem to determine whether any of these two bounds is tight. 


For query complexity, it is well known that $D(f) = O(C^2(f))$ 
(where $C(f)$ is the standard certificate complexity without the unambiguity requirement) \cite{BW}. Since $C(f)\leq D^{sc}(f)$ , we have $D(f) = O( \lr{D^{sc}(f)}^{2})$. 
On the other hand, Savicky \cite{Sav02} has constructed a function $f$ with $D(f) = \Omega( \lr{D^{sc}(f)}^{1.128...})$.

Recently, G\"o\"os, Pitassi and Watson \cite{GPW15} made progress on these longstanding open questions, 
by constructing a function $f$ with $D(f) = \tilde{\Omega}( \lr{D^{sc}(f)}^{1.5})$ and showing that 
a separation between $D(f)$ and $D^{sc}(f)$ can be ``lifted" from query complexity to the communication complexity.
Thus, $D^{cc}(f)=\tilde{\Omega}(\log^{1.5} \chi(f))$

We improve this result by constructing $f$ with $D(f)=\Omega( \lr{D^{sc}(f)}^{2-o(1)})$ (which almost matches
the upper bound of $D(f)=O( \lr{D^{sc}(f)}^{2})  $).
This also implies a similar separation between $D^{cc}(f)$ and $\log \chi(f)$,
via the lifting result of \cite{GPW15}.

Our construction is based on the {\em cheat-sheet} method that was very recently 
developed by Aaronson, Ben-David and Kothari \cite{ABK} to give better separations between query complexity in different models of computation
and related complexity measures. (In particular, they used cheat sheets to give the first superquadratic separation between
quantum and randomized query complexities, for a total Boolean function.) 

The cheat-sheet method takes a function $f$ and produces a new function $f_{CS}$ consisting of a composition of 
several copies of $f$, together with ``cheat-sheets" that allow a quick verification of values of $f$, given a pointer
to the right cheat sheet. In section \ref{sec:res}, we observe that cheat-sheet functions have the property that $f_{CS}^{-1}(1)$
can be partitioned into subcubes of dimensions that can be substantially smaller than $D^{sc}(f)$.

This property does not immediately imply a better separation between $D$ and $D^{sc}$, because the cheat-sheet construction
does not give a similar partition for $f_{CS}^{-1}(0)$. However, we can compose the cheat-sheet construction with several 
rebalancing steps which rebalance the complexity of partitions for $f_{CS}^{-1}(0)$ and $f_{CS}^{-1}(1)$.
Repeating this composed construction many times gives $D(f)=\Omega( \lr{D^{sc}(f)}^{2-o(1)})$.

\section{Preliminaries}

We use $[n]$ to denote $\{1, 2, \ldots, n\}$ and $\log n$ to denote $\log_2 n$. We use the big-$\tilde{O}$ notation which is a counterpart of 
big-$O$ notation that hides polylogarithmic factors:
\begin{itemize}
\item
$f(n)=\tilde{O}(g(n))$ if $f(n) = O(g(n) \log^c g(n))$ for some constant $c$ and
\item
$f(n)=\tilde{\Omega}(g(n))$ if $f(n)= \Omega(\frac{g(n)}{\log^c g(n)})$ for some constant $c$.
\end{itemize}

\subsection{Complexity measures for Boolean functions}

We study the complexity of Boolean functions $f:\{0, 1\}^n \rightarrow \{0, 1\}$
(or, more generally, functions $f:\Sigma_1 \times \Sigma_2 \times \ldots \times \Sigma_n \rightarrow \{0, 1\}$
where $\Sigma_i$ are finite sets of arbitrary size).
We denote the input variables by $x_1, \ldots, x_n$.
Then, the function is $f(x_1, \ldots, x_n)$.
Often, we use $x$ as a shortcut for $(x_1, \ldots, x_n)$ and $f(x)$ as a 
shortcut for $f(x_1, \ldots, x_n)$. 

$AND_n$ denotes the function $AND_n(x_1, \ldots, x_n)$ which is 1 if $x_1=\ldots=x_n=1$
and 0 otherwise. $OR_n$ denotes the function $OR_n(x_1, \ldots, x_n)$ which is 1 if $x_i=1$
for at least one $i$ and 0 otherwise.  

For functions $f:\{0, 1\}^n \rightarrow \{0, 1\}$ and $g:\Sigma_1 \times \ldots \Sigma_k \rightarrow \{0, 1\}$,
their composition is the function $f\circ g^n: (\Sigma_1 \times \ldots \Sigma_k)^n \rightarrow \{0, 1\}$
defined by  
\[ f\circ g^n (x_{11}, \ldots, x_{1k}, \ldots, x_{n1}, \ldots, x_{nk}) = 
f(g(x^{(1)}), \ldots, g(x^{(n)})) \]
where $x^{(i)}$ denotes $(x_{i1}, \ldots, x_{ik})$.

We consider the complexity of Boolean functions in the decision tree (query) model.
More information on this model can be found in the survey by Buhrman and de Wolf \cite{BW}.
A summary of recent results (that appeared after \cite{BW} was published) can be found in \cite{ABK}.

\smallskip

{\bf Deterministic query complexity.}
A {\em deterministic query algorithm} or {\em deterministic decision tree} is a deterministic 
algorithm $\A$ which accesses the input $(x_1, \ldots, x_n)$ by querying the input variables $x_i$. 
We say that $\A$ computes a function $f(x_1, \ldots, x_n)$ if, for any $(x_1, \ldots, x_n)\in\{0, 1\}^n$,
the output of $\A$ is equal to $f(x_1, \ldots, x_n)$. The {\em deterministic query complexity} of $f$, $D(f)$,
is the smallest $k$ such that there exists a deterministic query algorithm $\A$ that computes $f(x_1, \ldots, x_n)$ and, on every $(x_1, \ldots, x_n)$, queries at most $k$ input variables $x_i$.

{\bf Partial assignments.}
A \emph{partial assignment} in $ \Sigma_1 \times \Sigma_2 \times \ldots \Sigma_n $, where $ \Sigma_i $ are   finite alphabets, is a function   $a : I_a \to  \bigcup_{i=1}^n \Sigma_i $, where $ I_a \subset [n]   $, satisfying $ a(i) \in \Sigma_i $ for all $ i \in I_a$.
We say that an assignment $a$ {\em fixes a variable} $x_i$ if $i\in I_a$.
The length of a partial  assignment is the size of the set $I_a$.
A string $x \in  \Sigma_1 \times \Sigma_2 \times \ldots \Sigma_n $ is \emph{consistent} with the assignment $a$ if  $x_i =a(i)$ for all $ i \in I_a$; then we denote $ x \sim a$.   
Every partial assignment defines an associated \emph{subcube}, which is the set of all  strings consistent with that assignment:
\[ 
S_C = \lrb{x \in  \Sigma_1 \times \Sigma_2 \times \ldots \Sigma_n  \ \vline\  x \sim a }.
 \]  
In the other direction, for every subcube there is a unique partial assignment that defines it.

{\bf Certificate complexity.}
For $b\in \BC$, a \emph{$b$-certificate} for a function $f :  \Sigma_1 \times \Sigma_2 \times \ldots \Sigma_n \to\BC$ is a partial assignment $ a$ such that  $ f(x) = b$ for all $x\sim a$.
The \emph{$b$-certificate complexity} of $f$ (denoted by $ C_b(f) $) is the smallest number $k$ such that,
for every $x:f(x)=b$, there exists a $b$-certificate $a$ of length at most $k$ with $x\sim a$.
The \emph{certificate complexity} of $f$ (denoted by $ C(f) $) is the maximum of $ C_0(f) $, $ C_1(f) $.

Equivalently, we can say that $C_b(f)$ is the smallest number $k$ such that the set $f^{-1}(b)$ can be 
written as a union of subcubes $S_C$ corresponding to partial assignments of length at most $k$. 
 
{\bf Unambiguous certificate complexity.}
The \emph{unambiguous $b$-certificate complexity} of $f$  (denoted by $ UP_b(f) $) is the smallest number $k$ 
such that we can choose a collection of $b$-certificates $a_1, \ldots, a_m$ of length at most $k$ 
with the property that, for every $x:f(x)=b$, there is exactly one $a_i$ with $x\sim a_i$.
Equivalently, $UP_b(f)$ is the smallest number for which the set $f^{-1}(b)$ can be written as a disjoint union of subcubes defined by fixing at most $k$ variables. 

The \emph{deterministic subcube partition complexity} of $f$ (denoted by $ D^{sc}(f) $ \cite{Kothari}) 
is defined as maximum of $ UP_0(f) $, $ UP_1(f) $ and is the smallest $k$ for which $\Sigma_1 \times \Sigma_2 \times \ldots \Sigma_n$ can be written as a disjoint union of subcubes obtained by fixing at most $k$ variables, with $f$ being constant on every subcube.
$D^{sc}(f)$ is also known as {\em two-sided unambiguous certificate complexity} \cite{GPW15} and has been studied 
under several other names (from {\em non-overlapping cover size} \cite{BOH} to 
{\em non-intersecting complexity} \cite{Belovs}).

\subsection{Technical lemmas}

Let $ f (x_1, \ldots, x_n, y_1, \ldots, y_k) :  \BC^n \times [N]^k \to \BC $ 
be a function with some $\{0, 1\}$-valued variables and some variables which take values in a larger set $[N]$.

We can 'Booleanize' $f$ in a following way. Let $d=\lceil \log N \rceil$. We fix a mapping $h$ from $\BC^{d}$ onto $[N]$ and define 
$ \tilde f : \BC^{n + k d}  \to \BC$ with variables $x_i, i\in[n]$ and $y_{ij}, i\in[k], j\in [d]$ by
\[ \tilde f(x_1, \ldots, x_n, y_{11}, \ldots, y_{kd}) = f(x_1, \ldots, x_n, h(y_{11}, \ldots, y_{1d}), \ldots, h(y_{k1}, \ldots, y_{kd})) .\] 
Similarly to equation (3) in \cite{GPW15}, we have

\begin{lemma}\label{th:booleanization}
For any measure $ m \in \lrb{D,UP_0, UP_1, C_0, C_1} $,
\[ 
m(\tilde f ) \leq  m(f)  \cdot  \lceil \log  N \rceil.
\]
Moreover, $ D(f) \leq  D(\tilde f )   $.
\end{lemma}

\begin{proof}
In appendix \ref{app:tech}.
\end{proof}

A second technical result that we use is about combining partitions into subcubes for several sets $\Sigma_i^{n_i}$
into a partition for $\Sigma_1^{n_1}\times \Sigma_2^{n_2}\times \ldots \times \Sigma_m^{n_m}  $.

\begin{lemma}\label{th:product_partition}
	Suppose that $ \Sigma_1 $, \ldots, $ \Sigma_m $ are finite alphabets and $ n_1, \ldots, n_m \in \mbb N $. Suppose that for each $ i \in [m] $ there is a set $ K_i \subset \Sigma_i^{n_i} $ which can be partitioned into disjoint subcubes defined by assignments of length at most $ d_i  $.
	
	Then the set $ K_1 \times K_2 \times \ldots \times K_m \subset \Sigma_1^{n_1}\times \Sigma_2^{n_2}\times \ldots \times \Sigma_m^{n_m}  $ can be partitioned into disjoint subcubes defined by assignments of length at most $ d_1+\ldots+d_m $.
\end{lemma}

\begin{proof}
In appendix \ref{app:tech}.
\end{proof}

\subsection{Communication complexity}

In the standard model of communication complexity \cite{Yao79}, 
we have a function $f(x, y):X\times Y \rightarrow \{0, 1\}$,
with $x$ given to one party (Alice) and $y$ given to another party (Bob).
{\em Deterministic communication complexity} of $f$, $D^{cc}(f)$ is the smallest $k$
such that there is a protocol for Alice and Bob that computes $f(x, y)$ and, for any $x$ and $y$,
the amount of bits communicated between Alice and Bob is at most $k$.

The counterpart of $D^{sc}(f)$ for communication complexity is the partition number $\chi(f)$,
defined as the smallest $k$ such that there $X\times Y$ can be partitioned into $X_i\times Y_i$ for $i\in[k]$ 
so that each $(x, y)\in X\times Y$ belongs to exactly one of $X_i\times Y_i$ and,
for each $i\in [k]$, $f(x, y)$ is the same for all $(x, y)\in X_i\times Y_i$.

\section{Results}
\label{sec:res}

\begin{theorem} \label{th:main}
There is a sequence of functions $f$ with $D(f)=\Omega((D^{sc}(f))^{2-o(1)})$.
\end{theorem}

\begin{proof}[Proof sketch.]
We describe the main steps of the proof of Theorem \ref{th:main} (with each step described in more detail in the next subsection or in the appendix).
We use the {\em cheat-sheet} construction of Aaronson, Ben-David and Kothari \cite{ABK}
which transforms a given function $f$ into a new function $f^{t,c}\CS$ in a following way.

\begin{definition}\label{def:sh}
Suppose that $ \Sigma $ is a finite alphabet and $ f : \Sigma^n \to \BC $ satisfies $ C(f) \leq c \leq n$. Let $ t \in \mbb N $ be fixed.  A function $ f^{t,c}\CS : \Sigma^{tn} \times  [n]^{2^t tc}  \to \BC $ for an input $ (x,y) $,  $x\in  \Sigma^{t n}$,  $y\in  [n]^{2^t tc}   $  is defined as follows:
\begin{itemize}
\item the input $ x$ is interpreted as a $  t \times n$ matrix with entries $ x_{p,q} \in \Sigma $, $ p \in [t] $, $ q\in [n] $;
\item for $ p\in [t] $, the  $ p $th row of $ x $ is denoted by $ x^{(p)} = (x_{p,1},x_{p,2},\ldots,x_{p,n}) \in \Sigma^n$ and is interpreted as an input to the function $f$, with the value of function denoted $ b_p = f(x^{(p)})  \in \BC$, ;
\item the input $ y $ is interpreted as a three-dimensional array of size $ 2^t \times t \times c $   with entries $ y_{p,q,r}  \in  [n]$, $ p \in [2^t] $, $ q\in [t] $, $ r \in [c] $;
\item we denote $ \hat p =1+ \sum_{i=1}^{t} b_i \, 2^{i-1} \in [2^t] $;
\item the function  $ f^{t,c}\CS  $ is defined to be 1 for the input $ (x,y) $ iff for each $ q \in [t] $ the  following properties simultaneously hold: 
\begin{enumerate}
	\item the $ c $ numbers $ y_{\hat p, q, 1} $, \ldots, $y_{\hat p, q, c} $ are pairwise distinct, and 
	\item the assignment $ a^q : \lrb{y_{\hat p, q, 1} ,\ldots, y_{\hat p, q, c} } \to \Sigma$, defined by 
	$$  a^q(y_{\hat p,q,r}) = x_{q, y_{\hat p,q,r}}, \quad  \text{for all } r \in [c],  $$
	is a $ b_q $-certificate for $ f $.
\end{enumerate}
\end{itemize}
\end{definition}

The function $f^{t,c}\CS$ can be computed by computing $f(x^{(p)})$ for all $ p\in [t] $, determining $\hat p$ and 
verifying whether the two properties hold. Alternatively, if someone guessed $\hat p$, he could verify that $ b_p = f(x^{(p)})$ for all $p\in[t]$
by just checking that the certificates $a^q$ (which could be substantially faster than computing $f(x^{(p)})$.
This is the origin of the term ``cheat sheet" \cite{ABK} - we can view $ y_{\hat p,q,r}$ for $ q\in [t] $, $ r \in [c] $ 
as a cheat-sheet that allows to check $ b_p = f(x^{(p)})$ without running a full computation of $f$.

This construction has the following effect on the complexity measures $D(f)$, $C(f)$, $UP_0(f)$ and $UP_1(f)$:

\begin{lemma}\label{th:sh_properties}
Assume that $ f  $ and $ c $  satisfy the conditions of Definition \ref{def:sh}  and $ t \geq 2\log \lr{t D(f)} $. Then, we have:
\begin{enumerate}
\item $  \frac{t}{2} D(f)  \leq   D(f^{t,c}\CS)  \leq   t D(f) + tc$;
\item $ UP_1(f^{t,c}\CS) \leq 2tc$;
\item $ UP_0(f^{t,c}\CS) \leq t  D^{sc}(f)  + 2tc$;
\item $ C(f^{t,c}\CS) \leq 3tc  $.
\end{enumerate}
\end{lemma}

\begin{proof}
In section \ref{sec:lemmas}.
\end{proof}

We note that some of variables for the resulting function $f^{t,c}\CS$ take values in $[n]$, even if the original $f$ is Boolean.
To obtain a Boolean function as a result, we ``Booleanize" $f^{t,c}\CS$ as described in Lemma \ref{th:booleanization}.

For our purposes, the advantage of Lemma \ref{th:sh_properties} is that $UP_1(f^{t,c}\CS)$ for the resulting function $f^{t,c}\CS$ may be substantially smaller than $UP_1(f)$ (because $UP_1(f^{t,c}\CS)$ depends on $C(f)$ which can be smaller than $UP_1(f)$)!

However, going from $f$ to $f\CS$ does not decrease $UP_0$ (because $UP_0(f^{t,c}\CS)$ depends on $D^{sc}(f) $). 
To deal with that, we combine the cheat-sheet construction with two other steps which rebalance $UP_0$ and $UP_1$.
These two steps are composition with AND and OR which 
have the following effect on the complexity measures $D(f)$, $C(f)$, $UP_0(f)$ and $UP_1(f)$:

\begin{lemma}\label{th:AND_OR_composition}
	Suppose that $ g : \BC^N \to \BC $. Let $ f_{and} = AND_n \circ g^n  $ and $ f_{or } = OR_n \circ g^n  $. Then  
	\begin{align*}
	& D(f_{and}) = D(f_{or}) = n D(g); \\
	& C_0 (f_{or}) =n  C_0(g), \quad  C_1(f_{or}) =  C_1 (g); \\
	& C_0 (f_{and}) = C_0(g), \quad  C_1(f_{and}) = n C_1 (g); \\
	& UP_0 (f_{or})  \leq n UP_0 (g), \quad  UP_1(f_{or}) \leq   (n-1)  UP_0(g) +  UP_1 (g)   ; \\
	& UP_0 (f_{and}) \leq  (n-1)  UP_1(g) +  UP_0 (g), \quad  UP_1(f_{and}) \leq nUP_1 (g).
	\end{align*}
\end{lemma}

\begin{proof}
In appendix \ref{app:main}.
\end{proof}

By combining Lemmas \ref{th:booleanization}, \ref{th:sh_properties} and \ref{th:AND_OR_composition}, we get

\begin{lemma}\label{th:block}
Let 
$ f : \BC^N \to \BC  $
be fixed.

Suppose that
$ t,c,n \in \mbb N $ satisfy
 $ \max \lrb{nC_0(f), C_1(f)}  \leq c $ and $ t \geq  2\log \lr{ t D(f)} $.
Consider the function 
$$ f' :   \BC^{tNn^2 + 2^t tc n \lceil \log (Nn) \rceil }   \to \BC  $$
  defined as follows:
\[ 
f' =  AND_n  \circ   \widetilde {h}\CS^n,
\]
where 
$\widetilde {h}\CS $ is the boolean function associated to $ h^{t,c}\CS$ 
and
$ h:    \BC^{Nn}   \to \BC $ is defined as
\[ 
h:= OR_n \circ f^n  .
\]
Then the following estimates hold:
\begin{enumerate}
\item $ 0.5 tn^2 D(f) \leq  D(f') \leq   tn( c+n D(f)   )\lceil \log (Nn) \rceil $;
\item $ D^{sc}(f') \leq   tn\lr{2 c + D^{sc}(f) } \lceil \log (Nn) \rceil   $;
\item $ C_0(f')   \leq 3tc  \lceil \log (Nn) \rceil $;
\item $ C_1(f') \leq  3tcn   \lceil \log (Nn) \rceil $.
\end{enumerate}
\end{lemma}

\begin{proof}
In appendix \ref{app:main}.
\end{proof}

Applying Lemma \ref{th:block} results in a function $f'$ for which $D(f')$ is roughly $n^2 D(f)$ but $D^{sc}(f)$ is
roughly $n D^{sc}(f)$. We use Lemma \ref{th:block} repeatedly to construct a sequence of functions $f^{(1)}, f^{(2)}, \ldots$
in which $f^{(1)}=AND_n$ and each next function $f^{(m+1)}$ is equal to the function $f'$ obtained
by applying Lemma \ref{th:block} to $f^{(m)}=f$. The complexity of those functions is described by the
next Lemma.

\begin{lemma}\label{th:iteracija}
Let $ m \in \mbb N $. Then there are positive integers  $ a_0 $, $ a_1 $, $ a_2 $, $ a_3 $ s.t.  for all integers $ n \geq 2$ there exists  
$ N \in \mbb N $ and a function $ f^{(m)} : \BC^N \to \BC $ satisfying
\begin{align*} 
& N \leq a_0 \,  n^{9m}  \log^{10m-10} (n), \\
& a_1 \,  n^{2m-1} \log^{m-1}(n)  \leq  D(f^{(m)}) \leq a_2 \,  n^{2m-1} \log^{2m-2}(n), \\
&    D^{sc}(f^{(m)}) \leq a_3  \, n^{m} \log^{2m-2}(n), \\
&    C_0 (f^{(m)}) \leq a_3 \, n^{m-1} \log^{2m-2}(n), \\
&    C_1 (f^{(m)}) \leq a_3 \,  n^{m} \log^{2m-2}(n).
\end{align*}
\end{lemma}

\begin{proof}
In appendix \ref{app:main}.
\end{proof}

Theorem \ref{th:main} now follows immediately from Lemma \ref{th:iteracija} which implies that for every $ m\in \mbb N $ we have a family of functions satisfying
\[ 
D^{sc} (f) = \tilde O(n^m) , \quad  D(f ) = \tilde \Omega (n^{2m-1}).
\]
Then, $D(f)=\tilde\Omega((D^{sc}(f))^{2-\frac{1}{m}})$. Since this construction works for every $ m\in \mbb N $, we get
that $D(f)=\Omega((D^{sc}(f))^{2-o(1)})$ for an appropriately chosen sequence of functions $f$.
\end{proof}

{\bf Communication complexity implications.}
The standard strategy for transferring results from the domain of query complexity to communication complexity 
is to compose a function $f:\{0, 1\}^n \rightarrow \{0, 1\}$ with a function $g:X\times Y \rightarrow \{0, 1\}$, obtaining the function
\[ f\circ g^n (x_1, \ldots, x_n, y_1, \ldots, y_n) = f(g(x_1, y_1), \ldots, g(x_n, y_n)) .\]
We can then define $x=(x_1, \ldots, x_n)$ as Alice's input and $y=(y_1, \ldots, y_n)$ as Bob's input.
Querying the $i^{\rm th}$ variable then corresponds to computing $g(x_i, y_i)$ and, if we have a query algorithm
which makes $D(f)$ queries, we can obtain a communication protocol for $f\circ g^n$ with a communication 
that is approximately $D(f) D^{cc}(g)$. 

Building on an earlier work by Raz and McKenzie \cite{RM}, G\"o\"os, Pitassi and Watson \cite{GPW15} have shown
\begin{theorem}
\label{th:gpw}
\cite{GPW15}
For any $n$, there is a function $g:X\times Y$ such that, for any $f:\{0, 1\}^n \rightarrow \{0, 1\}$, we have
\begin{enumerate}
\item
$D^{cc}(f\circ g^n) = \Theta(D(f) \log n)$ and
\item
$\log \chi(f\circ g^n) = O(D^{sc}(f) \log n)$.
\end{enumerate} 
\end{theorem}

In this theorem, $D^{cc}(f\circ g^n) = O(D(f) \log n)$ and $\log \chi(f\circ g^n) = O(D^{sc}(f) \log n)$ follow easily 
by replacing queries to variables $x_i$ with computations of $g$. The part of the theorem that is not obvious (and is quite difficult technically) is that one also has $D^{cc}(f\circ g^n) = \Omega(D(f) \log n)$, i.e. communication protocols for $f\circ g^n$ cannot be better than the ones obtained from deterministic query algorithms for $f$.

By combining Theorems \ref{th:main} and \ref{th:gpw}, we immediately obtain

\begin{corollary}
There exists $h$ with $D^{cc}(h)=\Omega(\log^{2-o(1)} \chi(h))$.
\end{corollary}

\subsection{Proof of Lemma \ref{th:sh_properties}}
\label{sec:lemmas}

\begin{proof}[Proof of Lemma \ref{th:sh_properties}]
	For brevity, we now denote $f^{t,c}\CS$ as simply $f\CS$.
	As before, the first $ tn $ variables are indexed by  pairs $ (p,q)  $, where $ p \in [t] $, $ q\in [n] $; the remaining $ 2^t tc $   variables are indexed by  triples $ (p,q,r)  $, where $ p \in [2^t] $, $ q\in [t] $, $ r \in [c] $.  
	We denote $ x^{(p)} = (x_{p,1},x_{p,2},\ldots,x_{p,n}) \in \Sigma^n$ for each $ p \in [t] $.

	\textbf{Lower bound on deterministic complexity.} 
	Let $\A$ be an adversarial strategy of setting the variables $x_1, \ldots, x_n$
	in the function $f(x_1, \ldots, x_n)$ that forces an algorithm to make at least 
	$D(f)$ queries. We use $t$ copies of this strategy $\A_1, \ldots, \A_t$, with $
	\A_p$ setting the values of $x_{p, 1}, \ldots, x_{p, n}$ in $f(x_{p, 1}, \ldots, x_{p, n})$. 
	The overall adversarial strategy is 
	\begin{enumerate}
		\item
		If a variable $x_{p, q}$ is queried and its value has not been set yet, use $\A_p$ 
		to choose its value;
		\item
		If a variable $y_{p, q, r}$ is queried and its value has not been set yet:
		\begin{enumerate}
			\item
			Let $b$ be the $q^{\rm th}$ bit of $p$. 
			\item
			Choose a certificate $a$ that certifies $f(x_{1}, \ldots, x_{n})=a$ and contains all 
			variables $x_i$ with indices $i=y_{p, q, r'}$ for which we have already chosen values.
		`	If possible, choose $a$ so that, in addition to those requirements, $a$ is also
			consistent with the values among $x_{q, 1}, \ldots, x_{q, n}$ that we have have already fixed.
			\item
			Set $y_{p, q, r}$ be an index of a variable that is fixed by $a$ but is not among
			$y_{p, q, r'}$ that have already been chosen.
			\item
			If $x_{q, y_{p, q, r}}$ is not already fixed, fix it using the adversarial 
			strategy $\A_q$.
		\end{enumerate}
	\end{enumerate}
	
	If less than $\frac{t D(f)}{2}$ queries are made, less than $\frac{t}{2}$ of 
	$f(x^{(p)})$ for $p\in[t]$ have been fixed. Thus, more than
	$\frac{t}{2}$ of $f(x^{(p)})$ are not fixed yet. 
	This means that there are more than $2^{\frac{t}{2}} > t D(f)$ possible choices for $\hat{p}$ that are consistent with the answers to queries that have been made so far.
	Since less than $\frac{t D(f)}{2}$ queries have been made, one of these choices
	has the property that no $y_{\hat{p}, q, r}$ has been queried for any $q$ and $r$.
	
	We now set the remaining variables $x_{p, q}$ so that $f(x^{(p)})$ equals the $p^{\rm th}$ bit of $\hat{p}$. Since no $y_{\hat{p}, q, r}$ has been queried, we are free to set them so that the correct requirements are satisfied for every $q$ ($y_{\hat{p}, q, r}$ are all distinct and $a^q(y_{\hat{p}, q, r})=x_{q, y_{\hat{p}, q, r}}$) or so that they are not satisfied.
	Depending on that, we can have $f\CS=0$ or $f\CS=1$.
	
	\textbf{Upper bound on deterministic complexity.} 
	We first use $tD(f)$ queries to compute $f(x_{p, 1}, \ldots, x_{p, n})$ for all $p\in[t]$. Once we know $f(x^{(p)})$ for all $p\in [t]$, we can
     calculate $\hat{p}$ from Definition \ref{def:sh}. We then use $tc$ queries to query $y_{\hat{p}, q, r}$ for all $q\in[t]$ and $r\in [c]$
     and $tc$ queries to query $x_{q, y_{\hat{p}, q, r}}$ for all $q\in[t]$ and $r\in [c]$. Then, we can check all the conditions of Definition \ref{def:sh}.
	
	\textbf{Unambiguous 1-certificate complexity.} 
	
	We show $UP_1(f\CS )\leq  2tc$ by giving a mapping from 1-inputs $(x, y)$ to 1-certificates $a$,
	with the property that, for any two 1-inputs $(x, y)$ and $(x', y')$, the corresponding
	1-certificates $a$ and $a'$ are either the same or contradict one another in some variable.
	Then, for every 1-input $(x, y)$ there is exactly one certificate $a$ with $(x, y)\sim a$.
	
	To map a 1-input $(x, y)$ to a 1-certificate $a$, we do the following:
	\begin{enumerate}
		\item
		Let $ b_i =f( x^{(i)}) $, for each $ i\in [t] $, and $\hat{p}=1+\sum_{i=1}^t b_i 2^{i-1}$. 
		Since $ f\CS(x,y)=1 $, the $ c $ numbers $ y_{\hat p, i, 1} $, \ldots, $y_{\hat p, i, c} $ are  distinct and the assignment $ a_i : \lrb{y_{\hat p, i, 1} ,\ldots, y_{\hat p, i, c} } \to \Sigma$, defined by 
		$$  a_i(y_{\hat p,i,r}) = x_{i, y_{\hat p,i,r}}, \quad  \text{for all } r \in [c],  $$
		is a valid $ b_i $-certificate for $ f $.
		We define that $a$ must fix all variables fixed by $a_1, \ldots, a_t$ in the same way (i.e., $ a $ fixes $ x_{i,j} $ with indices $ j =y_{\hat p, i,r}  $, for all $ i\in [t] $ and $ r\in [c] $).
		\item
		 We define that $a$ also fixes the variables $y_{\hat{p}, q, r}$, $ q \in [t] $, $ r\in [c] $, to the values
		that they have in the input $(x, y)$.
	\end{enumerate}
	In each stage $ tc $ variables are fixed. Thus, the length of 
	the resulting partial assignment $a$ is  $2tc$.
	Notice that $ a $ is a  1-certificate for $ f\CS $, since the $ t $ assignments  $ a^1 $, \ldots, $ a^t $ uniquely determine $ \hat p $, and all variables $ y_{\hat p,q,r} $ are fixed to valid values, ensuring that $ f\CS= 1 $.

	We now show that, for any two different 1-certificates $a$, $a'$ constructed through this process, there is no input $(x, y)$ that satisfies both of them. 
	If $a$ and $a'$ differ in the part that is fixed in the first stage, there are two possible cases:
	\begin{enumerate}
		\item There exists $j\in [t]$ such that $ a_j $ is 0-certificate for $ f $,  whereas $ a'_j $ is a 1-certificate for $ f $ (or vice-versa). Then there must be no $x^{(j)}$ that satisfies both of them.
		\item For every $q\in [t]$, $ a_q $ and $ a'_q $ are both $ b_q $-certificates, for the same $ b_q \in \BC $ but there exists $j\in [t]$ such that $a_j$ differs from $a'_j$. In this case $ a_1 $, \ldots, $ a_t $ and $ a'_1 $, \ldots, $ a'_t $ determine the same value $ \hat p $.
		
Since $ a_j $ and $ a'_j $ are different certificates that fix the same variables (namely, they both fix 
$ x_{j,l} $ with indices $ l =y_{\hat p, j,r}  $, $r\in [c]$), they must fix at least one variable to different values. 
Then, no $ x^{(j)} $ can satisfy both $ a_j $ and $ a'_j $.
	\end{enumerate}

	If $a$ and $a'$ are the same in the part fixed in the 1st stage, the values of the 
	variables that we fix in the 1st stage
	uniquely determine $\hat{p}$ and, hence, the indices of variables that are fixed in the 2nd stage.
	Hence, the only way how $a$ and $a'$ could differ in the part fixed in the 2nd stage is if the same variable $y_{\hat{p}, q, r}$ is fixed to different values in $a$ and $a'$. This means that there is no $(x, y)$ that satisfies
	both $a$ and $a'$.

	\textbf{Unambiguous 0-certificate complexity.} 
	
	Similarly to the previous case, we give a mapping from 0-inputs $(x, y)$ to 0-certificates $a$,
	with the property that, for any two 0-inputs $(x, y)$ and $(x', y')$, the corresponding
	0-certificates $a$ and $a'$ are either the same or contradict one another in some variable.
	To map a 0-input $(x, y)$ to a 0-certificate $a$, we do the following:
	\begin{enumerate}
		\item
		We fix a partition of $\{0, 1\}^n$ into subcubes corresponding to assignments of length at most $D^{sc}(f)$
        with the property that $f$ is constant on every subcube.
		For each $p\in[t]$, $(x_{p, 1}, \ldots, x_{p, n})$ belongs to some subcube in 
		this partition.
		Let $a_p$ be the certificate that corresponds to this subcube. 
		We define that $a$ must fix all variables fixed by $a_1, \ldots, a_t$ in the same way.
		\item
		Let $\hat{p}=1+\sum_{i=1}^t b_i 2^{i-1}$ where $b_i = f(x^{(i)})$. (We note that $a_1, \ldots, a_t$
		determine the values of $f(x^{(1)}), \ldots, f(x^{(t)})$. Thus, $\hat{p}$ is determined by the variables that
		we fixed at the previous stage.) We define that $a$ also fixes the variables $y_{\hat{p}, q, r}$ to the values
		that they have in the input $(x, y)$.
		\item
		If there is $  q \in [t] $ such that the set $ \lrb{y_{\hat p, q,1}, \ldots, y_{\hat p, q,c}} $ is not equal to the set of variables fixed in some $b_q$-certificate (this includes the case when these $ c $ numbers are not distinct),
		the values of variables fixed by $a$ imply that $f\CS(x, y)=0$. In this case, we stop.
		\item
		Otherwise, there must be $ q' \in [t]  $ such that the set $ \lrb{y_{\hat p, q',1}, \ldots, y_{\hat p, q',c}} $
		is equal to the set of variables fixed in some $b_{q'}$-certificate but the actual assignment $x_{y_{\hat p, q',1}}, \ldots, x_{y_{\hat p, q',c}}$ is not a valid $b_{q'}$-certificate.
		
		In this case, we define that $a$ also fixes $x_{y_{\hat p, q, r}}$ for all 
		$q\in [t], r\in [c]$ to the values that they have in the input $(x, y)$.  
		Since this involves fixing $x_{y_{\hat p, q',1}}, \ldots, x_{y_{\hat p, q',c}}$
		which do not constitute a valid $b_{q'}$-certificate, this implies $f\CS(x, y)=0$.
	\end{enumerate}
	The first stage fixes at most $t D^{sc}(f)$ variables, the second stage fixes $tc$ variables and the last stage 
	also fixes $tc$ variables (some of which might have already been fixed in the 1st stage). Thus, the length of 
	the resulting 0-certificate $a$ is at most $t D^{sc}(f)+2tc$.
	
	We now show that, for any two different 0-certificates $a$, $a'$ constructed through this process, there is no input $(x, y)$ that satisfies both of them. If $a$ and $a'$ differ in the part that is fixed in the first stage, 
	there exists $j\in [t]$ such that
	$a_j$ differs from $a'_j$. Since $a_j, a'_j$ correspond to different subcubes in a partition,
	there must be no $x^{(j)}$ that satisfies both of them.
	
	If $a$ and $a'$ are the same in the part fixed in the 1st stage, the values of the 
	variables that we fix in the 1st stage
	uniquely determine $\hat{p}$ and, hence, the indices of variables that are fixed in the 2nd stage.
	Hence, the only way how $a$ and $a'$ could differ in the part fixed in the 2nd stage is if the same variable $y_{\hat{p}, q, r}$ is fixed to different values in $a$ and $a'$. This means that there is no $(x, y)$ that satisfies
	both $a$ and $a'$.
	
	If $a$ and $a'$ are the same in the part fixed in the first two stages, the values of the variables  
	that we fix in these stages uniquely determine the indices of variables that are fixed in the last stage and the same argument applies.

	\textbf{Certificate complexity.} 
	Since $ b $-certificate complexity is no larger than unambiguous $ b $-certificate complexity, we immediately conclude that
	\[ 
	C_1 (f\CS) \leq UP_1 (f\CS) \leq 2tc.
	\]

	We show $C_0(f\CS )\leq 3tc$ by giving a mapping from 0-inputs $(x, y)$ to 0-certificates $a$ (now different certificates are not required to contradict one another).
	Then, the collection of all 0-certificates $a$ to which some $(x, y)$ is mapped covers
	$f\CS^{-1}(0)$ (possibly with overlaps).
	
	To map a 0-input $(x, y)$ to a 0-certificate $a$, we do the following:
	\begin{enumerate}
		\item
		Let $a_p$, for each $p\in[t]$, be a $ f(x^{(p)}) $-certificate  that is satisfied by $x^{(p)} = (x_{p, 1}, \ldots, x_{p, n})$. 
		We define that $a$ must fix all variables fixed by $a_1, \ldots, a_t$ in the same way.
		\item
		Let $\hat{p}=1+\sum_{i=1}^t b_i 2^{i-1}$ where $b_i = f(x^{(i)})$. (Notice that  $\hat{p}$ is determined by the variables that
		we fixed at the previous stage.) We define that $a$ also fixes the variables $y_{\hat{p}, q, r}$ to the values
		that they have in the input $(x, y)$.
		\item
		If there is $  q \in [t] $ such that the set $ \lrb{y_{\hat p, q,1}, \ldots, y_{\hat p, q,c}} $ is not equal to the set of variables fixed in some $b_q$-certificate, 
		the values of variables fixed by $a$ imply that $f\CS(x, y)=0$. In this case, we stop.
		\item
		Otherwise, there must be $ q' \in [t]  $ such that the set $ \lrb{y_{\hat p, q',1}, \ldots, y_{\hat p, q',c}} $
		is equal to the set of variables fixed in some $b_{q'}$-certificate but the actual assignment $x_{y_{\hat p, q',1}}, \ldots, x_{y_{\hat p, q',c}}$ is not a valid $b_{q'}$-certificate.
		
		In this case, we define that $a$ also fixes $x_{y_{\hat p, q, r}}$ for all 
		$q\in [t], r\in [c]$ to the values that they have in the input $(x, y)$.  
		Since this involves fixing $x_{y_{\hat p, q',1}}, \ldots, x_{y_{\hat p, q',c}}$
		which do not constitute a valid $b_{q'}$-certificate, this implies $f\CS(x, y)=0$.
	\end{enumerate}
	The first stage fixes at most $t c$ variables, the second stage fixes $tc$ variables and the last stage 
	also fixes $tc$ variables (some of which might have already been fixed in the 1st stage). Thus, the length of 
	the resulting 0-certificate $a$ is at most $3tc$.
	We   conclude that
	$ C_0(f\CS) \leq 3tc  $
	and also
	\[ 
	C(f\CS) \leq 3tc .
	\]

\end{proof}

\section{Conclusions}

A deterministic query algorithm induces a partition of the Boolean hypercube $\{0, 1\}^n$
into subcubes that correspond to different computational paths that the algorithm can take.
If $\A$ makes at most $k$ queries, each subcube is defined by
values of at most $k$ input variables. 

It is well known that one can also go in the opposite direction, with a quadratic loss. 
Given a partition of $\{0, 1\}^n$ into subcubes $S_i$ defined by 
fixing at most $k$ input variables with a function $f$ constant on every $S_i$, one can construct
a query algorithm that computes $f$ with at most $k^2$ queries \cite{JLV}. 

In this paper, we show that this transformation from partitions to algorithms is close to being optimal,
by exhibiting a function $f$ with a corresponding partition for which any 
deterministic query algorithm requires $\Omega(k^{2-o(1)})$ queries.
Together with the ``lifting theorem" of \cite{GPW15}, this implies a similar result 
for communication complexity: there is a communication problem $f$ for which the input set can be
partitioned into $2^k$ rectangles with $f$ constant on every rectangle but
any deterministic communication protocol needs to communicate 
$\Omega(k^{2-o(1)})$ bits.

An immediate open question is whether randomized or quantum algorithms (protocols) still require 
$\Omega(k^{2-o(1)})$ queries (bits). It looks plausible that the lower bound for deterministic query complexity
$D(f)$ for our construction can be adapted to randomized query complexity, with a constant factor loss every
time when we iterate our construction. If this is indeed the case, 
we would get a similar lower bound for randomized query algorithms. With randomized communication
protocols, the situation is more difficult because the $D^{cc}(f\circ g^n) = \Theta(D(f) \log n)$ result
of \cite{GPW15} has no randomized counterpart \cite{GJ+15}. 

In the quantum case, our composed function $f$ no longer requires $\Omega(k^{2-o(1)})$ queries 
because one could use Grover's quantum search algorithm \cite{Grover} to evaluate $AND_n$ and $OR_n$. 
Using this approach, we can show that the function $f^{(m)}$ of Lemma \ref{th:iteracija} can be computed with $O(n^{m-\frac{1}{2}})$ quantum queries
which is less than our bound on $D^{sc}(f)$. Generally, it seems that we do not know functions $f$ for which quantum query complexity
$Q(f)$ is asymptotically larger than $D^{sc}(f)$.

\begin{appendix}

\newpage
{\Huge Appendix}

\section{Proofs of technical lemmas}
\label{app:tech}

\begin{proof}[Proof of Lemma \ref{th:booleanization}]
If we have a query algorithm for $f$, we can replace each query querying a variable $y_i$ by $d=\lceil \log N\rceil$ queries querying
variables $y_{i1}, \ldots, y_{id}$ (and a computation of $h(y_{i1}, \ldots, y_{id}))$) and obtain an algorithm for $\tilde{f}$. 
Conversely, we can transform an algorithm computing $\tilde{f}$ into an algorithm computing $f$ by replacing a query to $y_{ij}$ by a 
query to $y_i$.

For certificate complexity measures ($C_a$ and $UP_a$, $a\in\{0, 1\})$, let
\[ x = (x_1, \ldots, x_n, h(y_{11}, \ldots, y_{1d}), \ldots, h(y_{k1}, \ldots, y_{kd})) \]
be the input to $f$ corresponding to an input $\tilde{x}=(x_1, \ldots, x_n, y_{11}, \ldots, y_{kd})$ to $\tilde{f}$. 
We can transform a certificate for the function $f$ on the input $x$ into a certificate for the function $\tilde{f}$ on the input $\tilde x$
by replacing each variable $y_i$ with $d=\lceil \log N\rceil$ variables $y_{i1}, \ldots, y_{id}$. This gives 
$C_a(\tilde f ) \leq  C_a(f)  \cdot  \lceil \log  N \rceil$. 

If the certificates for $f$ with which we start are unambiguous, then, for any two different certificates $I_1, I_2$, 
there is a variable in which they differ. If $I_1$ and $I_2$ differ in one of $x_i$'s, the corresponding certificates for $\tilde f$ differ 
in the same $x_i$. If $I_1$ and $I_2$ differ in one of $y_i$'s, the corresponding certificates for $\tilde f$ differ in at least one of $y_{ij}$
for the same $i$ and some $j\in [d]$. This gives
$UP_a(\tilde f ) \leq  UP_a(f)  \cdot  \lceil \log  N \rceil$. 
\end{proof}

\begin{proof}[Proof of Lemma \ref{th:product_partition}]
	For each $ i\in [m] $ we can express $ K_i $ as
	\[ 
	K_i  = \bigcup_{j=1}^{t_i}  S_{i,j},
	\]
	where $ S_{i,j}  \subset \Sigma_i^{n_i}$ are subcubes and for all $ j, j' \in [t_i] $ with $j\neq j'$ the subcubes are disjoint, i.e., $ S_{i,j} \cap S_{i,j'}  = \emptyset$.  Let $ a_{i,j}  : I_{i,j} \to \Sigma_i $, $ I_{i,j} \subset [n_i] $, $ \lrv{I_{i,j}} \leq d_i $, be the partial assignment, associated to the subcube $ S_{i,j} $, $ j \in [t_i] $, $ i \in [m] $.
	
	Let $ J = [t_1] \times [t_2] \times \ldots \times [t_m] $ and denote $ \mb j = (j_1, \ldots, j_m) \in J $.    For each $ k\in [m] $ denote $ N_k = n_1+n_2+\ldots+n_k $.
	Define a set $ S_{\mb j} $ for each $ \mb j \in J $ as
	\[ 
	S_{\mb j} = S_{1,j_1} \times S_{2,j_2} \times \ldots \times S_{m,j_m}.
	\]
	Fix any $  {\mb j} \in J$.
	Clearly, $ S_{\mb j}  \subset \Sigma_1^{n_1}\times \Sigma_2^{n_2}\times \ldots \times \Sigma_m^{n_m}   $.

	Denote $ I + k = \lrb{i + k \ \vline\  i \in I} $ for a  set $ I \subset \mbb R $ and $ k \in \mbb R $. 
	Notice that $ S_{\mb j} $ is a subcube, with the associated partial assignment $ A_{\mb j}  : I_{\mb j} \to  \bigcup_{i \in [m]} \Sigma_i $, where the set $ I \subset [N_m]  $ is defined as 
	\[ 
	I = I_{1,j_1} \cup  \lr{I_{2,j_2} + N_1}\cup  \lr{I_{3,j_3} +N_2} \cup \ldots \cup  \lr{I_{m,j_m} + N_{m-1}}
	\]
	and
	\[ 
	A_{\mb j}  (i) = 
	\begin{cases}
	a_{1,j_1}(i) , & i \in   I_{1,j_1} \subset   [N_1]  ,\\
	a_{2,j_2}(i-N_1) , &  i \in   I_{1,j_1}+N_1 \subset   [N_1+1 \tdots N_2] ,\\
	\ldots \\
	a_{k,j_k}(i-N_{k-1}) , & i \in I_{k,j_k} + N_{k-1}\subset  [N_{k-1}+1 \tdots N_k],\\
	\ldots \\
	a_{m,j_m}(i-N_{m-1}) , & i \in I_{m,j_m} + N_{m-1}\subset  [N_{m-1}+1 \tdots N_m] .
	\end{cases}
	\]
	We also observe that the length of the assignment $A_{\mb j}$ defining $ S_{\mb j} $ is
	\[ 
	\lrv{I_{\mb j}} = \lrv{I_{1,j_1}}+\lrv{I_{2,j_2}}+ \ldots + \lrv{I_{m,j_m}} \leq d_1 + d_2+ \ldots+ d_m.
	\]
	
	
	Finally, $ S_{\mb j} $, $ \mb j \in J $ define a partition of the  whole space $ \Sigma_1^{n_1}\times \Sigma_2^{n_2}\times \ldots \times \Sigma_m^{n_m}  $
into disjoint subcubes. To see that, fix any $ x=(x_1,\ldots, x_m) $ with $ x_i \in \Sigma_i^{n_i} $, $ i\in [m] $. Since $ S_{i,j} $, $ j\in [t_i] $, partition the space $ \Sigma_i^{n_i} $, there is a unique $ j_i \in [t_i] $ such that $ x_i \in \Sigma_i^{n_i} $, for all $ i \in [m] $. But then  there is a unique $ \mb j = (j_1,\ldots, j_m) \in J $ such that $ x \in S_{\mb j} $.
	
	
\end{proof}

\section{Proofs of Lemmas \ref{th:AND_OR_composition}-\ref{th:iteracija}}
\label{app:main}

\begin{proof}[Proof of Lemma \ref{th:AND_OR_composition}]
	The equalities $  D(f_{and}) = D(f_{or}) = n D(g)  $ immediately follow from \cite[Lemma 3.1]{Tal2013}.
	The equalities $ C_0 (f_{or}) =n  C_0(g) $ and $  C_1(f_{or}) =  C_1 (g) $ have been shown in \cite[Proposition 31]{Gilmer2013}.
	
	Since $ f_{and}= \neg OR_n \circ (\neg g)^n $, also $ C_0 (f_{and}) = C_0(g) $ and $  C_1(f_{and}) = n C_1 (g) $ follow from \cite[Proposition 31]{Gilmer2013}. It remains to show the upper bounds on
	$ UP_0 (f_{or})   $ and $ UP_1(f_{or})  $  (the proof for $ f_{and} $ is similar).

	Let $ UP_0 (g) = u_0 $ and $ UP_1 (g) = u_1 $. Then, for each $ b\in \BC $, $ g^{-1}(b) $ can be partitioned into disjoint subcubes defined by assignments of length at most $ u_b $.
	
	For an input $x\in \BC^{Nn}$, let $x=(x^{(1)}, \ldots, x^{(n)}), x^{(i)}\in \BC^N$.
	We have $ f_{or}(x)  =0 $  iff $ g(\tilde x^{(k)}) = 0 $ for all $ k \in [n] $. Hence
	\[ 
	f_{or}^{-1} (0) = \lr{g^{-1}(0)}^n.
	\]
	By Lemma \ref{th:product_partition}, $ f_{or}^{-1} (0)  $ can be partitioned into disjoint subcubes  defined by assignments of length at most $ n u_0 $. Thus, $ UP_0 (f_{or})  \leq n UP_0 (g) $.

For $f_{or}^{-1} (1)$, we have
	\[ 
	f_{or}^{-1} (1) = \bigcup_{j=1}^n K_j , 	K_j = \lr{g^{-1}(0)}^{j-1} \times g^{-1}(1) \times \BC^{(n-j)N}  \subset \BC^{Nn} .
	\] 
	We note that the sets $ K_j $ are disjoint (since each $K_j$ consists of all $ x \in \BC^{Nn} $ for which
$j$ is the smallest index with $g(x^{(j)})=1$). 
	By  Lemma \ref{th:product_partition}, $\lr{g^{-1}(0)}^{j-1} \times g^{-1}(1)$ can be partitioned into disjoint subcubes defined by assignments of length at most $ (j-1) u_0 + u_1$. This induces a partition of $K_j$ into subcubes
defined by the same assignments.
	
Hence, $  f_{or}^{-1} (1)  $ can be partitioned into disjoint subcubes defined by assignments of length at most $ (n-1)u_0 + u_1 $. That is, $ UP_1(f_{or}) \leq   (n-1)  UP_0(g) +  UP_1 (g)  $.

	In particular, this implies that $ D^{sc}(f_{or}) \leq n D^{sc} (g) $; note that this also follows from \cite[Proposition 3]{Kothari}.
\end{proof}

\begin{proof}[Proof of Lemma \ref{th:block}]
By Lemma \ref{th:AND_OR_composition},
\begin{align*}
&   D(h) =  D(OR_n) \cdot D(f) = n D(f),\\  
&    D^{sc}(h) \leq   D^{sc}(OR_n) \cdot D^{sc}(f) = n D^{sc}(f),\\ 
&    C_0(h) =n C_0(f)  \leq c,\\ 
&    C_1(h) = C_1(f)  \leq c . 
\end{align*}
By Lemma \ref{th:sh_properties}, 
\begin{align*}
& D(h^{t,c}\CS)  \geq  0.5D(h) = 0.5  tn D(f), \\
& D(h^{t,c}\CS)  \leq    t D(h) + tc= tn D(f)  + tc, \\
& UP_1(h^{t,c}\CS) \leq 2tc, \\
& UP_0(h^{t,c}\CS) \leq t  D^{sc}(h)  + 2tc \leq  t n D^{sc}(f)  + 2tc , \\
& C(h^{t,c}\CS) \leq 3tc.
\end{align*}
By Lemma \ref{th:booleanization},
\begin{align*}
&  0.5 tn  D(f) \leq  D(\widetilde {h}\CS)  \leq    t(n D(f)  +c) \lceil \log (Nn) \rceil,  \\
&  UP_1(\widetilde {h}\CS) \leq 2tc  \lceil \log (Nn) \rceil,  \\
& UP_0(\widetilde {h}\CS) \leq (  t n D^{sc}(f)  +2 tc)  \lceil \log (Nn) \rceil , \\
& C(\widetilde {h}\CS) \leq 3tc   \lceil \log (Nn) \rceil  .
\end{align*}
Finally, again by Lemma \ref{th:AND_OR_composition},
\begin{align*}
& D(f') = n D(\widetilde {h}\CS), \\
& C_0(f')  =  C_0(\widetilde {h}\CS) , \\
&  C_1(f')  =  nC_1(\widetilde {h}\CS) , \\
& UP_1(f') \leq  n UP_1(\widetilde {h}\CS) , \\
& UP_0(f') \leq (n-1)UP_1(\widetilde {h}\CS) + UP_0(\widetilde {h}\CS).
\end{align*}
Consequently, we have the following estimates:
\begin{align*}
&    0.5 tn^2 D(f) \leq  D(f') \leq     tn(n D(f)  +c) D(f)   \lceil \log (Nn) \rceil, \\
& C_0(f')   \leq  3tc  \lceil \log (Nn) \rceil, \\
&  C_1(f') \leq  3tcn   \lceil \log (Nn) \rceil , \\
&  UP_1(f') \leq 2tcn   \lceil \log (Nn) \rceil , \\
&  UP_0(f') \leq \lr{2 tcn+ t n D^{sc}(f) } \lceil \log (Nn) \rceil  .
\end{align*}
The last two inequalities imply  $ D^{sc}(f')  = \max \lrb {UP_1(f'),UP_0(f')}  \leq   \lr{2 tcn+ t n D^{sc}(f) } \lceil \log (Nn) \rceil  $, concluding the proof.
\end{proof}

\begin{proof}[Proof of Lemma \ref{th:iteracija}]
The proof is by induction on $ m $. When $ m=1 $, take 	$ a_0 = a_1= a_2=a_3 =1 $. Then for any $ n \geq 2 $ one can choose $ N=n $, $  f^{(1)}  = AND_n $.

Suppose that for  $ m=b $ there are constants $ a_0 $, $ a_1 $, $ a_2 $, $ a_3 $ s.t. for all $ n \geq 2$ there is 
$ f^{(m)} : \BC^N \to \BC$, satisfying the given constraints. We argue that for $ m=b+1 $ there are positive integers
\begin{align*} 
& a_0' =   4 (a_2 + 4b)  a_0+ 128 a_2^4  a_3(a_2 + 4b)(a_0 +19b-10) , \\
&  a_1'  = 2 (2b-1)a_1 , \\
& a_2'  =  4(a_2+a_3)(a_2 + 4b)   (a_0 +19b-10)    , \\
&  a_3'  = 12 a_3  (a_2 + 4b )  (a_0 +19b-10) 
\end{align*}
s.t. for all $ n \geq 2 $ there is  $ N' \in \mbb N $ and $ f' : \BC^{N'} \to \BC $, satisfying  the necessary properties.

We fix   $ n \geq 2 $. Let $ f^{(b)}  $ be given by the inductive hypothesis. 
Let $ f^{(b+1)} =f'$ where $f'$ is obtained by applying Lemma \ref{th:block} to $ f=f^{(b)}$, with parameters $ {t= \lceil 4 \log \lr{ a_2 \,  n^{2b-1} \log^{2b-2}(n)}  \rceil +4}  $ and   $ {c=  \max \lrb{nC_0(f^{(b)}), C_1(f^{(b)})}   }$. 
Notice that for all $ D \geq 2 $ setting $ t \geq 4 \log (D) + 4  $ yields $ t  \geq 2\log (tD) $. Hence 
the chosen value of $ t $ satisfies  $ t \geq  2\log \lr{ t D(f^{(b)})} $ and
we have
$  f^{(b+1)} :  \BC^{N'}   \to \BC$, where  
\[ 
N'  = tNn^2 + 2^t tc n \lceil \log (Nn) \rceil ,
\]  
and,  by Lemma \ref{th:block}, 
\begin{align*}
& 0.5 tn^2 D(f^{(b)})     \leq D(f^{(b+1)})  \leq   tn  (c + nD(f^{(b)}))\lceil \log (Nn) \rceil ; \\
& D^{sc}(f^{(b+1)})  \leq     tn\lr{2 c + D^{sc}(f^{(b)}) } \lceil \log (Nn) \rceil   ; \\
& C_0(f^{(b+1)}) \leq 3tc \lceil \log (Nn) \rceil; \\
& C_1(f^{(b+1)})  \leq 3tcn \lceil \log (Nn) \rceil . 
\end{align*}

Notice that
\begin{align*}
& \lceil \log (Nn) \rceil \leq   (a_0 +19b-10)   \log  (n); \\
& t=4+ \lceil 4 \log \lr{a_2 \,  n^{2b-1} \log^{2b-2}(n)}  \rceil \leq  4 (a_2 + 4b)   \log  (n); \\
& t  > 4\log \lr{a_2 \,  n^{2b-1} \log^{2b-2}(n)}  \geq  4(2b-1 )\log (n);\\
& 2^t \leq  32  a_2^4 \,  n^{8b-4} \log^{8b-8}(n);\\
& c \leq  a_3 \,  n^{b} \log^{2b-2}(n) ; \\ 
& 2 c + D^{sc}(f^{(b)})  \leq 3a_3 \,  n^{b} \log^{2b-2}(n) ;\\
& c + nD(f^{(b)}) \leq   (a_2 + a_3)\,  n^{2b} \log^{2b-2}(n)
\end{align*}
We conclude that
\begin{multline*} 
N' 
\leq  
 4 (a_2 + 4b)   \log  (n) 
a_0 \,  n^{9b+2}  \log^{10b-10} (n )
+  \\  
32  a_2^4 \,  n^{8b-4} \log^{8b-8}(n) 
 4 (a_2 + 4b)   \log  (n)  
a_3 \,  n^{b} \log^{2b-2}(n) n 
(a_0 +19b-10)   \log  (n) 
=\\
 4 (a_2 + 4b)  a_0 \,  n^{9b+2}  \log^{10b-9} (n)
+
128 a_2^4  a_3(a_2 + 4b)(a_0 +19b-10)  n^{9b-3}\log^{10b-8}(n) \leq \\
a_0' \, n^{9b+9}  \log^{10b}(n).
\end{multline*}
Similarly, 
\begin{align*}
& D(f^{(b+1)})  \leq  
4(a_2 + 4b)  \log  (n) 
n  \,
 (a_2+a_3)\,  n^{2b} \log^{2b-2}(n)   
 (a_0 +19b-10)   \log  (n)  = 
a_2'  \, n^{2b+1} \log^{2b}(n), \\
& D(f^{(b+1)}) \geq   2(2b-1)   \log  (n)  n^2 \, a_1\,  n^{2b-1} \log^{b-1}(n) =   a_1' \,  n^{2b+1} \log^{b}(n) , \\
& D^{sc}(f^{(b+1)}) \leq 12 (a_2 + 4b )   \log  (n)    n \,   a_3  \, n^{b} \log^{2b-2}(n)  (a_0 +19b-10)   \log  (n)  = a_3'  \, n^{b+1} \log^{2b}(n),\\ 
& 	C_0(f^{(b+1)}) \leq   12(a_2 + 4b )   \log  (n) a_3 \,     n^{b} \log^{2b-2}(n)      (a_0 +19b-10)   \log  (n) =  a_3'  n^{b} \log^{2b}(n), \\
& C_1(f^{(b+1)}) \leq   12(a_2 + 4b )   \log  (n) a_3 \,  n^{b+1} \log^{2b-2}(n) (a_0 +11b-10)   \log  (n) =  a_3'  n^{b+1} \log^{2b}(n),
\end{align*}
completing the inductive step.
\end{proof}

\end{appendix}
\end{document}